\documentclass[journal]{IEEEtran}
\usepackage{mathrsfs}
\usepackage{amssymb} 
\usepackage[mathscr]{eucal}  
\usepackage{makecell}
\usepackage{cite}
\usepackage{graphicx}
\usepackage{psfrag}
\usepackage{subfigure}
\usepackage{url}
\usepackage{stfloats}
\usepackage{amsmath}
\usepackage{float}
\usepackage{array}
\usepackage{amsthm}
\usepackage{booktabs}  %
\usepackage{algorithm} 
\usepackage{algorithmic} 
\usepackage{color}
\usepackage{cases}
\usepackage{bm}
\usepackage{changepage}
\usepackage{amsmath,amsfonts,amssymb,amsbsy,paralist,ifthen}
\allowdisplaybreaks[4]

\newtheorem{Pro}{Proposition}
\newtheorem{The}{Theorem}

\begin{document}
\title{Average AoI in Pinching Antenna-assisted WPCNs with Probabilistic LoS Blockage}

\author{Huimin Hu, Ruihong Jiang, Yanqing Xu,~\IEEEmembership{Member,~IEEE}, Jiarui Ma, Fang Fang,~\IEEEmembership{Senior Member,~IEEE} \\ 
 \thanks{
H. Hu is with the School of Communications and Information Engineering, Xi'an University of Posts and Telecommunications, Xi'an 710121, China. (e-mail: huiminhu@xupt.edu.cn)

R. Jiang is with the State Key Laboratory of Networking and Switching Technology, Beijing University of Posts and Telecommunications, Beijing 100876, China (e-mail: rhjiang@bupt.edu.cn).

Y. Xu is with the School of Science and Engineering, The Chinese University of Hong Kong, Shenzhen, 518172, China (email: xuyanqing@cuhk.edu.cn).

J. Ma is with the XUPT-UNH College of Information Engineering, Xi'an University of Posts and Telecommunications (XUPT), Xi'an 710121, China. (email: 180536547@qq.com)

F. Fang is with the Department of Electrical and Computer Engineering and the Department of Computer Science, Western University, London, ON N6A 3K7, Canada. (e-mail: fang.fang@uwo.ca).
}
}
\maketitle

\begin{abstract}
This paper analyzes the age of information (AoI) for a pinching antenna (PA)-assisted wireless powered communication network (WPCN) with probabilistic line-of-sight (LoS) blockage. 
AoI is a key metric for evaluating the freshness of status updates in IoT networks, and its optimization is crucial for ensuring the performance of time-critical applications.
{To facilitate analysis and gain useful insights, we consider a representative scenario,} where an IoT device harvests energy from a {base station (BS) equipped with a PA} and transmits data packets to it. In particular, a PA is deployed along the waveguide to achieve directional wireless power transmission (WPT) and wireless information transmission (WIT), thereby enhancing energy harvesting (EH) efficiency and communication reliability.
A ``harvest-then-transmit" protocol is designed for the IoT device. The IoT device harvests energy via the PA until its capacitor is fully charged, then transmits status updates using all stored energy. 
We derive closed-form expressions for the average AoI by analyzing the capacitor charging time, transmission success probability, and inter-arrival time of successful updates. To minimize the average AoI, we formulate an optimization problem of PA position, {and propose a one-dimensional search to solve it.} 
{The simulation results show that the optimal PA position is the one closest to the IoT device, and this conclusion can be extended to the multi-IoT devices frequency division multiple access (FDMA) scenario. {The PA-based systems significantly outperform the conventional fixed-antenna systems.}
}
\end{abstract}


\section{Introduction} \label{Introd}
With the extensive growth of Internet of Things (IoT) applications (e.g., smart monitoring, industrial automation), the demand for real-time status updates has driven the research on age of information (AoI), which is a critical metric quantifying the freshness of received data \cite{AoI_app1}. 
Meanwhile, wireless powered communication networks (WPCNs) have risen as a sustainable approach to mitigate the energy limitations of IoT devices, where dedicated energy transmitters provide wireless power via radio frequency (RF) signals \cite{WPCN_app1, WPCN_app2}. 
However, traditional WPCNs face key challenges in AoI optimization, especially under time-varying blocking conditions. 
The dynamic changes of obstructions or the movement of users can seriously degrade the transmission efficiency of RF signals, not only causing a significant drop in the energy collection rate of IoT devices, but also prolong the charging time of capacitors, thereby leading to a significant increase in AoI.

To tackle these challenges, flexible and reconfigurable antenna technologies have undergone remarkable advancements in recent years.
Among these, pinching antennas (PAs), with their unique architecture of radiating elements coupled to dielectric waveguides, have become a potential solution for flexible antenna design \cite{PA_intro1, PA_intro1add, PA_intro2}.  
Unlike traditional fixed antennas, PAs enable the radiating element's movement across an extensive spatial range along the waveguide, permitting real-time adjustment of the antenna's position within the communication area \cite{PA_intro3,PA_intro4}. 
This dynamic reconfiguration capability enables wireless networks to proactively respond to environmental changes, enhancing communication reliability and information freshness.


Recently, some studies have explored on the system performance of PA-assisted WPCNs \cite{PA_WPCN1, PA_WPCN2, PA_WPCN3, PA_WPCN4}.
For instance, in \cite{PA_WPCN1}, the sum rate of a novel PA-assisted WPCN was maximized by activating multiple PAs along a waveguide, which helps establish reliable line-of-sight (LoS) links with multiple wireless nodes.
In \cite{PA_WPCN2},
maximizing the minimum data rate of a wireless powered PA network (WPPAN) was considered, where a single waveguide integrated with multiple PAs was adopted to tackle the intrinsic doubly near-far issue in wireless powered networks.
In \cite{PA_WPCN3}, the average achievable rate of a flexible PA-enabled simultaneous wireless information and power transfer (SWIPT) system was optimized, with closed-form expressions for the UE’s average harvested energy and average achievable rate being derived.
In \cite{PA_WPCN4}, the PA-assisted SWIPT system was studied,  which maximized the total power received by the UE by jointly optimizing information receiver's power allocation and PA positioning.

However, existing works on PA-assisted WPCNs \cite{PA_WPCN1, PA_WPCN2, PA_WPCN3, {PA_WPCN4}} only focus on sum rate or energy efficiency, and lack investigation on AoI.
As a core metric quantifying the freshness of data received by terminal devices, AoI directly determines the effectiveness of real-time IoT services (such as smart monitoring, industrial automation, and remote healthcare). In these scenarios, outdated data (i.e., high AoI) can lead to erroneous decision-making (e.g., delayed fault detection in industrial equipment) or degraded service quality (e.g., inaccurate real-time environmental monitoring), making AoI a non-negligible performance indicator. 

Therefore, this paper investigates the average AoI of a PA-assisted WPCN, where a PA is deployed along a waveguide to enable directional wireless power transfer (WPT) and wireless information transmission (WIT).
The contributions of this paper are as follows:
\begin{itemize}
\item We propose a novel PA-assisted WPCN tailored for AoI minimization. Specifically, by dynamically adjusting the position of the PA along a waveguide, the PA can enable directional WPT and directional WIT. 

\item We establish a probabilistic LoS blockage model for the PA-IoT link, and derive the probability of successful data packet transmission under this probabilistic LoS blockage model.

\item A closed-form expression is derived for average AoI by analyzing the capacitor charging time and transmission success probability. The simulation results show that {the PA systems substantially outperform the conventional fixed-antenna systems.}
\end{itemize}

\begin{figure}[]
\centerline
{\includegraphics[width = 3.1in]{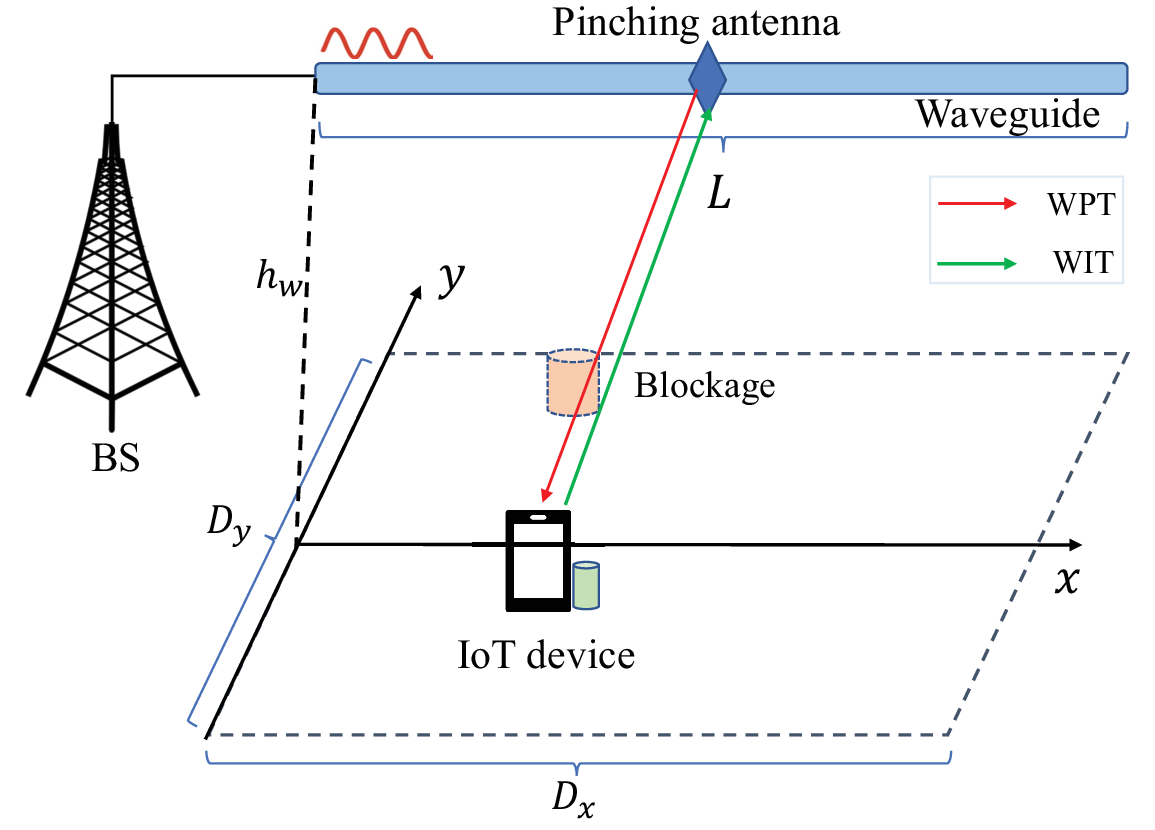}}
\caption{The system model.} 
\label{model_figure}
\end{figure}
\section{System Model} \label{model}
\subsection{System Model} \label{system} 

A PA-assisted WPCN with probabilistic LoS blockage is considered, where a PA is deployed along a waveguide of length $L$ to serve {multiple} IoT devices equipped with a single antenna.
The IoT devices are  equipped with a capacitor of size $B_{\rm max}$ and are located within a rectangular area $[0, D_x] \times [-\frac{D_y}{2}, \frac{D_y}{2}]$.
For the convenience of analysis, we consider a simple scenario with only one IoT device, as illustrated in Fig. \ref{model_figure}.
It is noteworthy that the single IoT device network model constitutes a basic component of complex networks. With the deployment of multiple IoT devices, frequency-division multiple access (FDMA) enables the decomposition of the complex network into multiple point-to-point networks identical to the one we analyzed.

The waveguide is located above the IoT device area, aligned along the $x$-axis at a height of $h_w$, and the feed point's position coordinate is denoted as $L_{\rm FP} = (0, 0, h_w)$.
The locations of the PA on the waveguide and the IoT device are $L_{\rm PA}=(x_p, 0, h_w)$ ($x_p \in [0, L]$) and $L_{\rm IoTD}=(x_{u}, y_{u}, 0)$, respectively.
To avoid interference between energy transfer and information transmission, the system adopts the Harvest-Then-Transmit (HTT) protocol to implement downlink WPT and uplink WIT.
First, the base station (BS) sends energy signals to the IoT device via the PA. Then, once the capacitor is fully charged, the IoT device uses the harvested energy to transmit information to the BS through the PA, and this cycle repeats.
It is assumed that the system time is divided into $T$ time slots, indexed by $t \in \mathcal{T} = \{1, 2, \ldots, T\}$, with each slot having a duration of $\Delta_T$.

\subsection{Channel Model with Probabilistic LoS Blockage} \label{channel} 
To capture scenarios where obstacles (such as buildings, trees, etc.) may exist in actual wireless environments, and the probabilistic nature of LoS links between antennas and devices caused by such obstacles, we consider LoS blockage scenarios.
It is assumed that the channel is constant in one time slot and varies independently in different time slots.
Let a Bernoulli random variable $\gamma(t) \in \{0,1\}$ denote whether an LoS link exists in time slot $t$, where $\gamma(t) = 1$ indicates the existence of a direct LoS link and $\gamma(t) = 0$ indicates its non-existence.
According to \cite{LoS_pro}, the probability of the existence of the LoS link is 
 \begin{equation}
{\rm Pr}{(\gamma(t)=1)} = e^{-\beta d(x_p)^2},
\label{h_eq1}
\end{equation}
where $\beta \in (0, 1]$ is the blockage density parameter, and $d(x_p)$ is the Euclidean distance between the PA and the IoT device, given by
 \begin{equation}
d(x_p) = \sqrt{(x_p-x_u)^2 + y_u^2 + h_w^2},
\label{d_eq1}
\end{equation}
 
For LoS links ($\gamma(t) = 1$, $\forall t$), the PA's directional channel power gain follows the free-space propagation model and is given by
\begin{equation}
|h_{\rm LoS}(t, x_p)|^2 = \frac{\eta}{d(x_p)^2},
\label{h_eq2}
\end{equation}
where $\eta = (\frac{c}{4\pi f_c})^2$ accounts for PA's antenna gain and RF-to-DC conversion efficiency, 
$f_c$ and $c$ denote the carrier frequency and speed of light, respectively.
For non-line-of-sight (NLoS) links ($\gamma(t) = 0$, $\forall t$), the channel power gain is negligible, so we set $|h_{\rm NLoS}(t, x_p)|^2 = 0$. 

Therefore, the channel coefficient with probabilistic LoS blockage in time slot $t$ can be expressed as
\begin{equation}
h(t, x_p) = \gamma(t) h_{\rm LoS}(t, x_p).
\label{h_eq3}
\end{equation}


\subsection{EH Model and Communication Model} \label{channel} 
We consider the linear energy harvesting model. The harvested energy of the IoT device in time slot $t$ is given by
\begin{equation}
E(t, x_p) =  \eta_{\rm eh} P_w |h(t,x_p)|^2 \Delta_T, 
\label{eh_eq1}
\end{equation}
where $\eta_{\rm eh} \in (0,1)$ represents the energy conversion efficiency.

Since the LoS link exists probabilistically and the energy conversion efficiency $\eta_{\rm eh}$ is less than 1, the IoT device may need several time slots to harvest and accumulate energy to fully charge the capacitor.
Suppose that it takes $k$ time slots to fully charge, then it satisfies:
\begin{equation}
\sum_{i=1}^{k-1}E(i, x_p) < B_{\rm max} \le \sum_{i=1}^{k}E(i, x_p). 
\label{eh_eq2}
\end{equation}

%

Correspondingly, the achievable throughput in time slot $t$ can be expressed as
\begin{equation}
R(t,x_p) = B \Delta_T {\rm log}_2 \Big(1 +  \frac{B_{\rm max} |h(t,x_p)|^2}{\Delta_T  \sigma^2} \Big),
\label{data_trans_eq1}
\end{equation}
where $B$ and $\sigma^2$ are the system bandwidth and the additive white Gaussian noise, respectively.
The probability of successful transmission of a data packet in time slot $t$ can be expressed as
\begin{equation}
\begin{split}
p_{\rm s} (t,x_p)  = {\rm Pr} \Big(R(t, x_p) \ge D ),
\label{data_trans_eq2}
\end{split}
\end{equation}
where $D$ is the size of data packet. 

\begin{The}\label{suc_pro}
The probability of successful transmission of a data packet in time slot $t$ is given by
\begin{equation}
\begin{split}
p_{\rm s} (x_p)  = e^{-\beta d(x_p)^2} \cdot \mathbb{I}\left( d(x_p) \le \sqrt{\frac{B_{\rm max} \eta}{(2^\theta-1)\Delta_T \sigma^2}} \right), 
\label{data_trans_eq3}
\end{split}
\end{equation}
where $\theta =  \frac{D}{ B \Delta_T}$ and $\mathbb{I}(\cdot)$ denotes the indicator function, which equals 1 when the event in parentheses occurs and 0 otherwise.
\begin{proof}
A data packet can only be successfully transmitted under the LoS channel condition. 
Therefore, through the total probability formula, we can obtain the probability of successful transmission is
\begin{equation}
\begin{split}
p_{\rm s} (t, x_p)  & = {\rm Pr} \big(R(t, x_p) \ge D | \gamma(t)=1 \big)  {\rm Pr}\big(\gamma(t)=1\big) \\
                            & ={\rm Pr} \left({\rm log}_2 \left(1 +  \frac{B_{\rm max} \eta }{\Delta_T  \sigma^2 d(x_p)^2} \right) \ge \theta \right)  e^{-\beta d(x_p)^2}\\
                            & = {\rm Pr} \left( d(x_p) \le \sqrt{\frac{B_{\rm max} \eta}{(2^\theta-1)\Delta_T \sigma^2}} \right) e^{-\beta d(x_p)^2}\\
                            & = e^{-\beta d(x_p)^2} \cdot \mathbb{I}\left( d(x_p) \le \sqrt{\frac{B_{\rm max} \eta}{(2^\theta-1)\Delta_T \sigma^2}} \right). \nonumber
\label{data_trans_eq4}
\end{split}
\end{equation}
Since the probability is independent of the time slot $t$, $t$ is omitted. 
The proof of this Theorem is complete.
\end{proof}
\end{The}

%

\begin{figure}[]
\centerline
{\includegraphics[width = 3.2in]{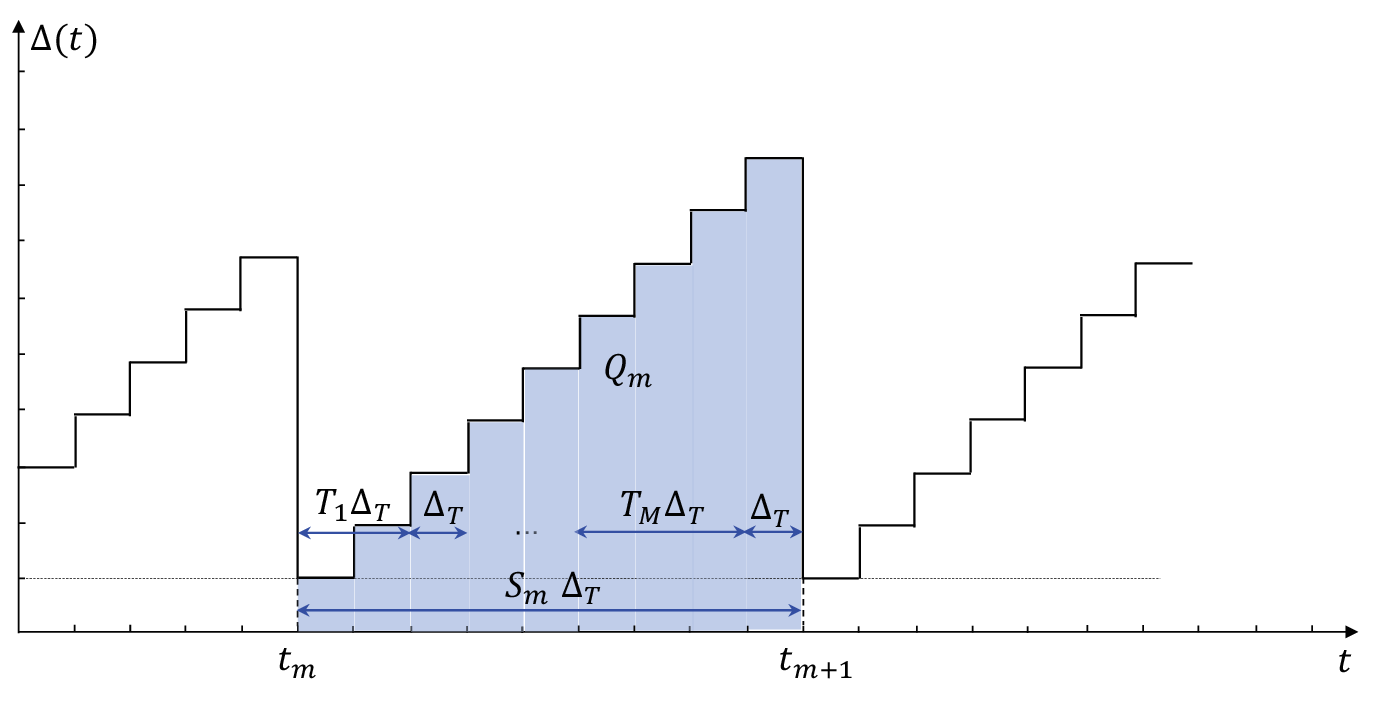}}
\caption{An example of AoI dynamic evolution. $S_m$ denotes the interval time between two consecutive received data packet, $T_m$ is the time between two consecutive capacitor’s recharges, $Q_m$ is the area under $\Delta(t)$ corresponding to the $m$-th received data packet.}
\label{AoI_model}
\end{figure}

\subsection{AoI model} \label{AoI}
According to \cite{AoI_prop}, AoI is defined to characterize the elapsed time since the last received data was generated. Therefore, the system's AoI in time slot $t$ is 
\begin{equation}
A(t) = t - U(t),
\label{aoi_eq1}
\end{equation}
where $U(t)$ is the generation time of the most recently received data packet.

Fig. \ref{AoI_model} presents an example of the AoI evolution of the considered system. When the capacitor is fully charged and the IoT device generates a data packet for transmission in the next time slot. If the data packet transmission is successful, the AoI at the BS will reset to $\Delta_T$; otherwise, the AoI will increase by $\Delta_T$. 
Specifically, $t_m$ and $t_{m+1}$ respectively represent the time when the $m$-th and the $(m+1)$-th data packet are successfully transmitted.
$S_m$ denotes the interval time between two consecutive received data packet, and $S_m \Delta_T = (t_{m+1} - t_m)$. Let $T_m$ denote the number of time slots required to fully charge.
We have $S_m = \sum\nolimits_{m=1}^M(T_m + 1)$, where $M$ is a discrete random variable that denotes the number of the data packet transmissions until successful receiving.

\section{Analysis of the Average AoI} \label{AoI_solution}
In this section, we analyze the performance of the PA-assisted WPCN with probabilistic LoS blockage in terms of the average AoI.

For a period of $T$ time slots with $M$ successful transmissions, the average AoI can be expressed as
\begin{equation}
\begin{split}
A_T &= \frac{1}{T} \sum_{t=1}^T A(t) = \frac{1}{T} \sum_{m=1}^M Q_m = \frac{M}{T}\left(\frac{1}{M} \sum_{m=1}^M Q_m \right) \\
       &= \frac{M}{T} {\mathbb E(Q)}.
\end{split}
\label{aoi_eq2}
\end{equation}

When $T \to \infty$,  we can obtain $\lim_{T \to \infty} \frac{M}{T} = \frac{1}{\Delta_T\mathbb E(S)}$. Therefore, the time average AoI is given by
\begin{equation}
\bar A = \lim_{T \to \infty}  A_T = \frac{\mathbb E(Q)}{\Delta_T\mathbb E(S)}.
\label{aoi_eq3}
\end{equation}

As can be seen from Fig. \ref{AoI_model}, the area of $Q_m$ is calculated as follows:
\begin{equation}
Q_m= \frac{S_m (S_m +1)}{2} \Delta_T^2.
\label{aoi_eq4}
\end{equation}

The average AoI can be calculated as the average area of the $Q_m$, i.e.,
\begin{equation}
\mathbb E(Q)= \frac{1}{\Delta_T \mathbb E(S)} \cdot \frac{ \mathbb E(S^2) + \mathbb E(S)}{2 } \Delta_T^2.
\label{aoi_eq5}
\end{equation}

\begin{figure*}[t]
    \centering
    \small
    \begin{minipage}[t]{0.328\linewidth}
        \centering
        \includegraphics[width=1.1\linewidth]{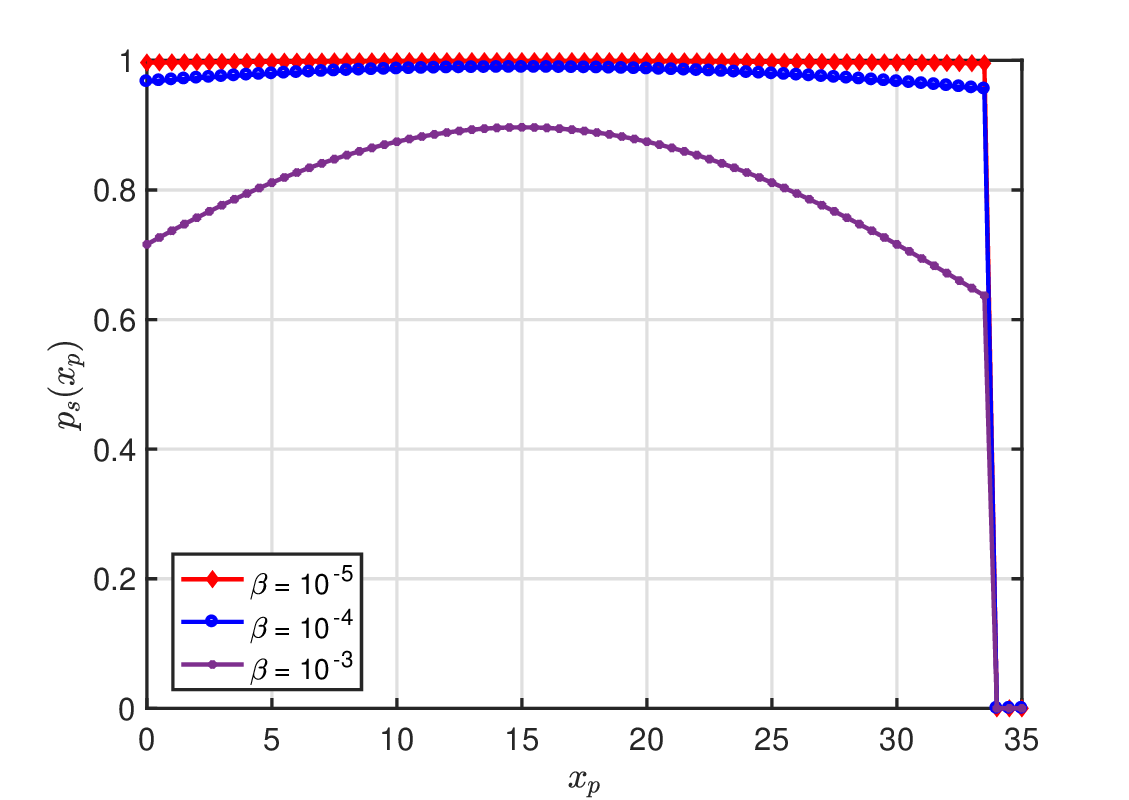}
        \caption{The probability of successful transmission of a data packet ($p_s{\rm}(x_p)$) versus the PA position $x_p$.}
        \label{ps_vs_xp}
    \end{minipage}
    \begin{minipage}[t]{0.328\linewidth}
        \centering
        \includegraphics[width=1.1\linewidth]{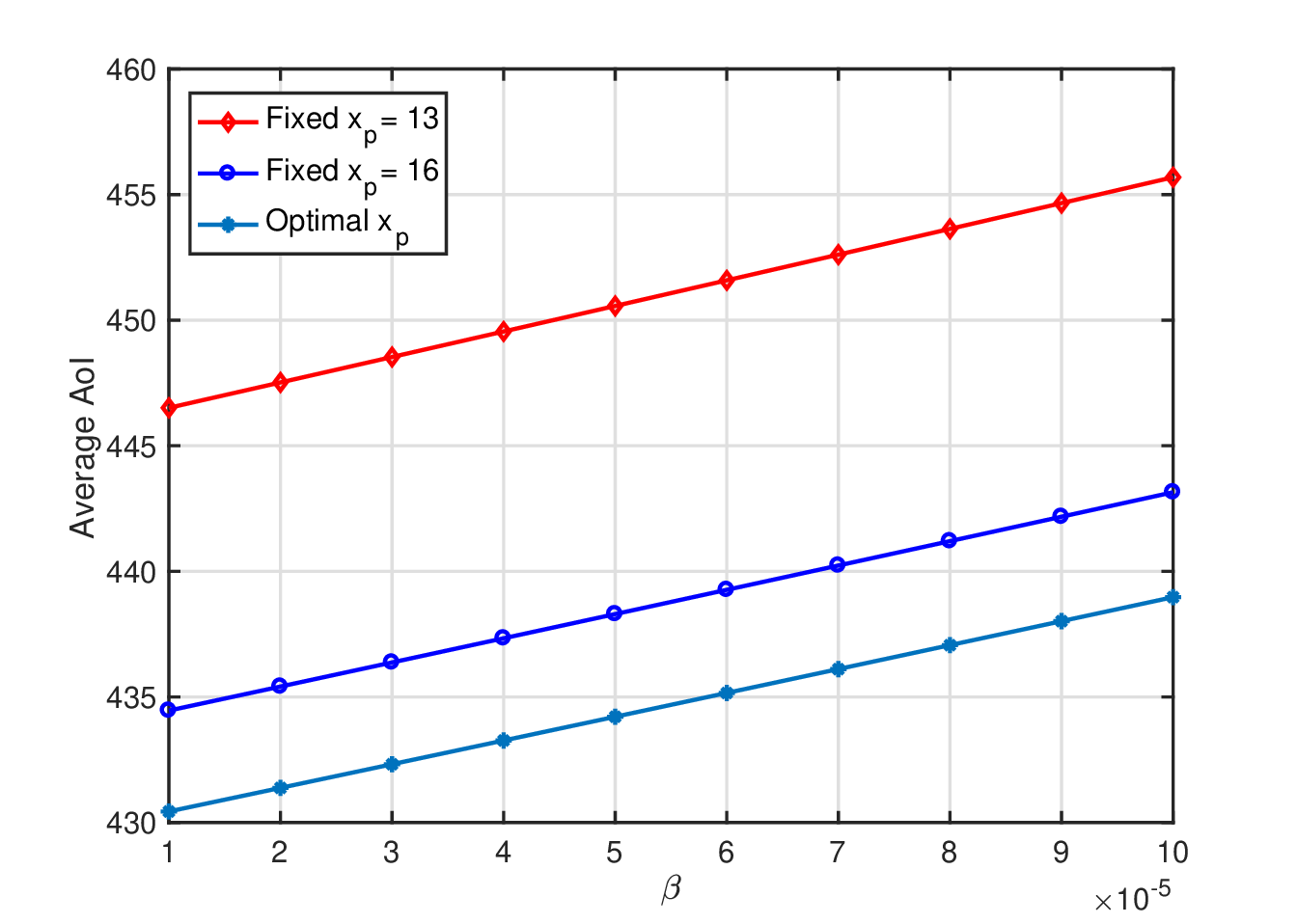}
        \caption{The average AoI versus the blockage parameter $\beta$ under different PA position $x_p$.}
        \label{AoI_vs_beta}
    \end{minipage}
        \begin{minipage}[t]{0.328\linewidth}
        \centering
        \includegraphics[width=1.1\linewidth]{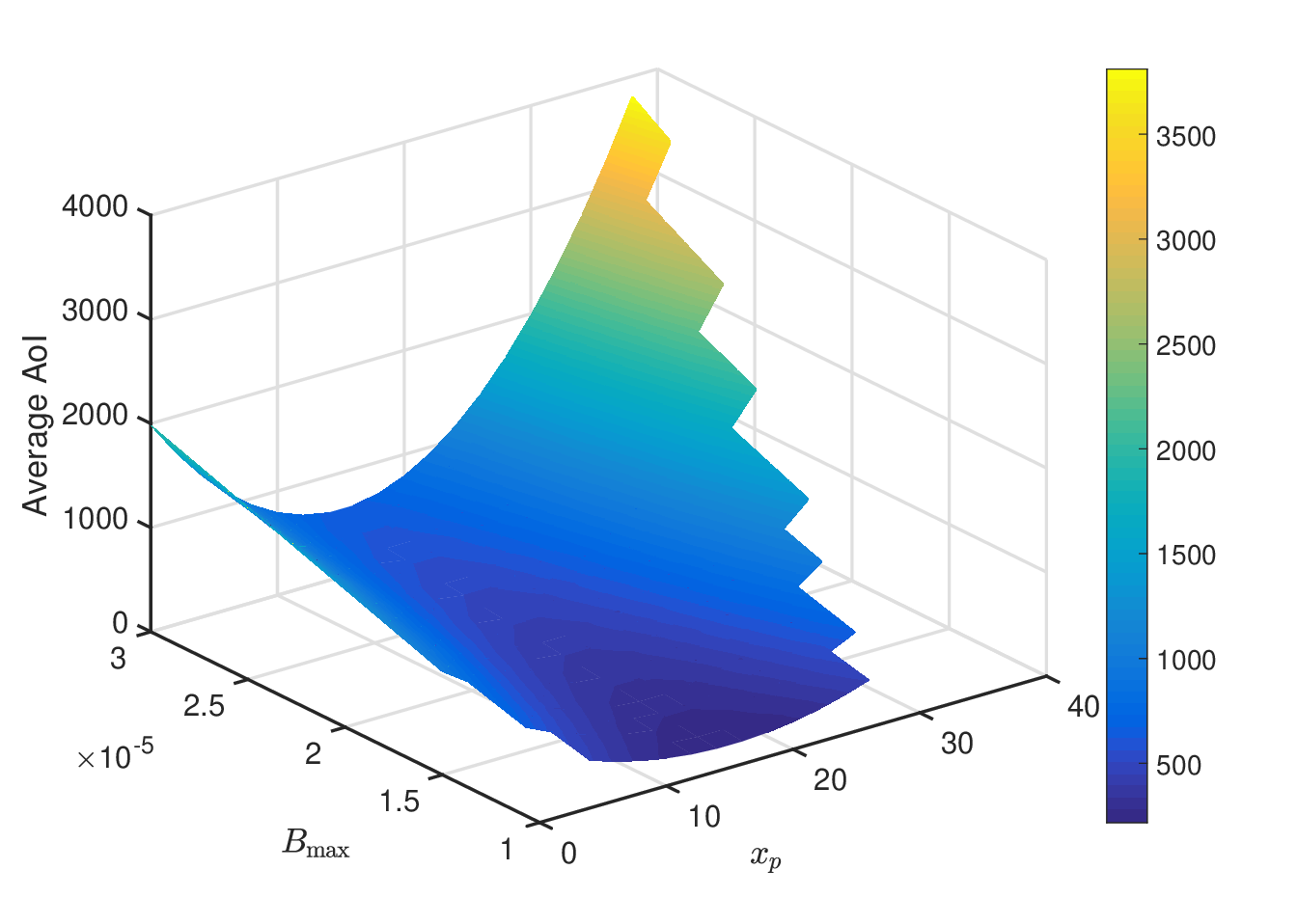}
        \caption{The average AoI versus the PA position $x_p$ and the capacitor of the IoT device $B_{\rm max}$.}
        \label{AoI_xp_Bmax}
    \end{minipage}
 \end{figure*}

\begin{Pro}\label{pro1}
The first-order and the second-order moments for the number of time slots between the successful transmission of two data packets, respectively, are given by
\begin{equation}
\mathbb E(S)= \frac{e^{-\beta d(x_p)^2} + \Big\lceil \frac{B_{\rm max} {d(x_p)^2}}{\eta_{\rm eh}P_w{\eta}\Delta_T } \Big\rceil }{e^{-\beta d(x_p)^2} p_{\rm s}(x_p)}.
\label{ES_eq1}
\end{equation}
and
\begin{equation}
\begin{split}
\mathbb E(S^2)&= \frac{\big\lceil \frac{B_{\rm max} {d(x_p)^2}}{\eta_{\rm eh}P_w{\eta}\Delta_T }\big\rceil(1-e^{-\beta d(x_p)^2}+\big\lceil \frac{B_{\rm max} {d(x_p)^2}}{\eta_{\rm eh}P_w{\eta}\Delta_T }\big\rceil)}{(e^{-\beta d(x_p)^2})^2 p_{\rm s}(x_p)}\\
&~~~+\frac{(e^{-\beta d(x_p)^2}+\big\lceil \frac{B_{\rm max} {d(x_p)^2}}{\eta_{\rm eh}P_w{\eta}\Delta_T }\big\rceil)^2(2-p_{\rm s}(x_p))}{(e^{-\beta d(x_p)^2})^2 (p_{\rm s}(x_p))^2}.
\end{split}
\label{ES_eq2}
\end{equation}
\end{Pro}
\begin{proof}
See Appendix \ref{App_Pro1}.
\end{proof}

\begin{The}\label{aoi_ave}
The average AoI for the considered PA-assisted WPCN is given by
\begin{equation}
\begin{split}
\bar A  &= \frac{ \Delta_T}{2} \left(\frac{\mathbb E(S^2)}{\mathbb E(S)}+1 \right) \\
   &= \frac{ \Delta_T}{2} \Bigg(\frac{\big\lceil \frac{B_{\rm max} {d(x_p)^2}}{\eta_{\rm eh}P_w{\eta}\Delta_T }\big\rceil(1-e^{-\beta d(x_p)^2}+\big\lceil \frac{B_{\rm max} {d(x_p)^2}}{\eta_{\rm eh}P_w{\eta}\Delta_T }\big\rceil)}{ e^{-\beta d(x_p)^2}(e^{-\beta d(x_p)^2}+\big\lceil \frac{B_{\rm max} {d(x_p)^2}}{\eta_{\rm eh}P_w{\eta}\Delta_T }\big\rceil)} \\
   &~~~+\frac{\Big(e^{-\beta d(x_p)^2}+\big\lceil \frac{B_{\rm max} {d(x_p)^2}}{\eta_{\rm eh}P_w{\eta}\Delta_T }\big\rceil \Big)(2-p_s(x_p))}{e^{-\beta d(x_p)^2 p_{\rm s}(x_p)}}+1 \Bigg).
\label{the_eq1}
\end{split}
\end{equation}
where the expression of $p_{\rm s}(x_p)$  is shown in Eq. \eqref{data_trans_eq3}.
\end{The}


To minimize the average AoI of the system, the optimal position of the PA is given by
\begin{equation}
\begin{split}
x_p^*= {\rm arg}\, \mathop{\rm min}\limits_{x_p \in [0,L]} \bar{A}
\label{aoi_eq6}
\end{split}
\end{equation}
Given the closed-form expression of average AoI (Eq. \eqref{the_eq1}), we can find the optimal $x_p^*$ through one-dimensional search.

\section{Numerical results}\label{result}
In this section, we present numerical results to analyze the average AoI performance of the considered system. 
The length of waveguide is set to $L=35$ m, the height of waveguide is set to $h_w = 10$ m.
The length of the rectangular area where the IoT device is located is set to $D_x = 35$ m, and the width is set to $D_y = 10$ m.
The IoT device is located at $(x_u, y_u, 0) = (10, 3, 0)$ m. 
The duration of a time slot is $\Delta_T = 1$ s. The carrier frequency is set to $f_c = 28$ GHz. The transmit power and the energy conversion efficiency are set to $P_w = 10$ W and $\eta_{\rm eh} = 0.7$, respectively. The battery capacity of the IoT device is $B_{\rm max} = 2^{-5}$ J. The data size is set to $D = 1000$ bits.
The system bandwidth is set to $B = 1000$ Hz, and the noise power is set to $\sigma^2 = -120$ dBm.  
The value range for LoS blockage parameter $\beta$ is $\{10^{-3}, 10^{-4}, 10^{-5}\}$.

Fig. \ref{ps_vs_xp} depicts the probability of successful transmission of a data packet versus the PA position $x_p$ under different LoS blockage parameter $\beta$. 
As the value of LoS blockage parameter $\beta$ increases, 
the probability of successful transmission of a data packet decreases. Additionally, the probability of successful transmission of a data packet varies with the position of the PA, being higher at locations closer to the IoT device. 
Moreover, the PA must not be positioned too far from IoT devices, otherwise the probability of successful transmission of a data packet becomes zero.

%
Fig. \ref{AoI_vs_beta} illustrates the average AoI versus the LoS blockage parameter $\beta$ under different PA position $x_p$ (Fixed $x_p$ and the optimal $x_p$ obtained through one-dimensional search.). 
First, it is evident that a smaller LoS blockage parameter $\beta$ results in a lower average AoI. A smaller $\beta$ implies a higher probability of maintaining LoS connections over a wider range of $x_p$. 
Moreover, the average AoI obtained under the optimal $x_p^*$ is the smallest, and the closer the fixed position is to the device position $x_u$, the smaller the AoI.

%
Fig. \ref{AoI_xp_Bmax} plots the average AoI versus the PA position $x_p$ and the capacitor capacity of the IoT device $B_{\rm max}$. 
As $x_p$ increases from 0 to 40, the system's average AoI first decreases and then increases, indicating that there exists an optimal $x_p^*$.
This optimal position strikes a balance between the energy harvesting efficiency from the PA and the reliability of LoS communication. When $x_p$ is too far from the IoT device, the probability of LoS links diminishes, leading to prolonged energy harvesting time and thus higher AoI.
Moreover, it can be observed that as $B_{\rm max}$ increases the system's average AoI increases. 
This is because an increase in battery capacity does not enhance the transmission success probability; instead, it takes longer to fully charge the capacitor, thus leading to an increase in the average AoI.
However, the capacitor capacity should not be too small; otherwise, even if the capacitor is fully charged in the LoS link yet the data transmission is unsuccessful, the AoI will continue to increase.

\section{Conclusion}\label{con}
This paper studied the AoI performance in a novel PA-assisted WPCN with probabilistic LoS blockage. A PA is deployed along a waveguide to enable directional WPT and WIT, which enhances EH efficiency and communication reliability. First, we established a probabilistic LoS blockage model for the PA-IoT link. Then, we derived a closed-form expression for the system's average AoI by analyzing the capacitor charging time, transmission success probability, and inter-arrival time of successful data packets. {The simulation results demonstrate that the PA system significantly outperforms the conventional fixed-antenna system, thanks to its capability to reconfigure the large-scale channel by effectively reducing free-space path loss and increasing the probability of maintaining a LoS link between the PA and the IoT device.}

\numberwithin{equation}{section}
\appendices 
\section{Proof of Proposition \ref{pro1}} \label{App_Pro1}
As mentioned earlier, $T_m$ denotes the number of time slots required to fully charge. Therefore, $T_m$ follows a negative binomial distribution.
That is, the distribution of the total number of time slots required to complete $K$ $(K=\frac{B_{\rm max}}{E_{\rm LoS}(t,x_p)})$ successful energy harvests, with the probability of successful energy collection in each time slot being $p$  $(p={\rm Pr}({\gamma(t)=1)})$.
Particularly, $E_{\rm LoS}(t,x_p)$ is the harvested energy of the IoT device in time slot $t$ under LoS links.
Therefore, we have
\begin{equation}
\begin{split}
\mathbb{E} [T] = \frac{K}{p} = \bigg\lceil \frac{B_{\rm max} {d(x_p)^2}}{\eta_{\rm eh}P_w{\eta}\Delta_T } \bigg\rceil e^{\beta d(x_p)^2}.
\label{app_eq1}
\end{split}
\end{equation}

The variance formula of the negative binomial distribution is ${\rm Var}(T)=\frac{K(1-p)}{p^2}$. Combining the relationship between variance and the second-order moment, we can obtain
\begin{equation}
\begin{split}
\mathbb{E} [T^2] &={\rm Var}(T) + [\mathbb{E}(T)]^2   \\
                &=\frac{K(1-p)}{p^2} + \bigg(\frac{K}{p} \bigg)^2\\
                & = \frac{ \big\lceil \frac{B_{\rm max} {d(x_p)^2}}{\eta_{\rm eh}P_w{\eta}\Delta_T }\big\rceil  \big(1-e^{-\beta d(x_p)^2}+ \big\lceil \frac{B_{\rm max} {d(x_p)^2}}{\eta_{\rm eh}P_w{\eta}\Delta_T } \big\rceil \big)}{(e^{-\beta d(x_p)^2})^2}
\label{app_eq2}
\end{split}
\end{equation}

Due to $S_m = \sum_{i=1}^{M}(T_i + 1)$, where $M$ represents the number of attempts required for the first successful transmission, following a geometric distribution with parameter $p_{\rm s}(x_p)$ (the probability of each successful transmission attempt is $p_{\rm s}(x_p)$ given in Eq. \eqref{data_trans_eq3}).
$T_m$ is the number of capacitor charging time slots before the $m$-th transmission attempt, and $T_1, T_2, \ldots, T_M$ is independent and identically distributed.

Take the expectation of $S_m$,  and use the law of total expectation (since $M$ is a random variable):
\begin{equation}
\begin{split}
\mathbb{E} [S] &= \mathbb{E} \Big[\mathbb{E} \big(\sum\nolimits_{i=1}^{M}(T_i + 1)|M \big) \Big] \\
              &\overset{(M=m)}{=} \mathbb{E}\Big[ \mathbb{E} \big(\sum\nolimits_{i=1}^{M}(T_i + 1)|M=m \big)\Big] \\
              &= \mathbb{E} \big(m(1+ \mathbb{E}(T))\big).
\label{app_eq3}
\end{split}
\end{equation}

The expectation of a geometric distribution is $\mathbb{E}(M)=1/p_{\rm s}(x_p)$. Therefore, we have 
\begin{equation}
\begin{split}
\mathbb{E} (S)&= (1+ \mathbb{E}(T))\mathbb{E}(M)\\
&= \Big(1 + \frac{K}{p}\Big)\cdot \frac{1}{p_{\rm s} (x_p)}\\
& =\frac{e^{-\beta d(x_p)^2} + \Big\lceil \frac{B_{\rm max} {d(x_p)^2}}{\eta_{\rm eh}P_w{\eta}\Delta_T } \Big\rceil }{e^{-\beta d(x_p)^2} p_{\rm s}(x_p)}.
\label{app_eq4}
\end{split}
\end{equation}

Taking the expectation of $M$ and combining it with the second-order moment of the geometric distribution, we obtain
\begin{equation}
\begin{split}
\mathbb{E}(M^2) =\frac{2-p_{\rm s}(x_p)}{(p_{\rm s}(x_p))^2}.
\label{app_eq5}
\end{split}
\end{equation}

Therefore, we can obtain that
\begin{equation}
\begin{split}
\mathbb E(S^2) &=\mathbb E(T^2)\mathbb E(M) +(1+\mathbb E(T))^2\mathbb E(M^2)\\
&=\frac{K(1-p+K)}{p^2 p_{\rm s}(x_p)}+\frac{(p+K)^2(2-p_s(x_p))}{p^2 p_{\rm s}(x_p)^2}\\
&=\frac{\big\lceil \frac{B_{\rm max} {d(x_p)^2}}{\eta_{\rm eh}P_w{\eta}\Delta_T }\big\rceil(1-e^{-\beta d(x_p)^2}+\big\lceil \frac{B_{\rm max} {d(x_p)^2}}{\eta_{\rm eh}P_w{\eta}\Delta_T }\big\rceil)}{(e^{-\beta d(x_p)^2})^2 p_{\rm s}(x_p)}\\
&~~~+\frac{(e^{-\beta d(x_p)^2}+\big\lceil \frac{B_{\rm max} {d(x_p)^2}}{\eta_{\rm eh}P_w{\eta}\Delta_T }\big\rceil)^2(2-p_{\rm s}(x_p))}{(e^{-\beta d(x_p)^2})^2 (p_{\rm s}(x_p))^2}.
\label{app_eq6}
\end{split}
\end{equation}


\end{document}